\documentclass[letterpaper]{article}
\usepackage[utf8]{inputenc}

\usepackage{doi}
\usepackage{xspace}
\usepackage{graphicx}
\usepackage{amsfonts}
\usepackage{amssymb}
\usepackage{amsthm} 

\newtheorem{theorem}{Theorem} 
\newtheorem{observation}{Observation}
\newtheorem{definition}{Definition}
\newtheorem{lemma}{Lemma}

\begin{document}

\title{Online Stacking with a Few Load/Unload Points}

\author{Martin Olsen\\Department of Business Development and Technology\\Aarhus University\\Denmark\\martino@btech.au.dk\\ }

\maketitle

\begin{abstract}
We consider the stacking problem where items are temporarily stored in stacks in a stacking area. The stacks are working as Last In First Out (LIFO) data structures. The objective is to avoid {\em shifts} or {\em restows} that occur when an item has to leave the stacking area earlier than an item above it. We present a simple online algorithm for the problem where we pick a stack for an incoming item without any information on future items. We present a sufficient condition for the algorithm to avoid shifts expressed as an inequality involving the dimension of the stacking area, the number of load/unload points and the maximum number of items present at the same time. The condition is relatively tight in the sense that we can find instances requiring shifts for {\em any} algorithm (including offline algorithms) with small modifications of the condition. Our results indicate that our algorithm is useful if the number of load/unload points is relatively small compared to the number of items.
\end{abstract}

\section{Introduction}

Stacking is an important operation within logistics. The items that are stacked can be containers in a container terminal or on a container ship~\cite{vanTwiller2024} or the items could be steel plates~\cite{Rei13}. A stack can be any type of Last In First Out (LIFO) data structure so a stack does not necessarily need to work in the conventional vertical manner. A train track where trains only enter and leave in one end can thus also be seen as a stack~\cite{Demange2012} and a lane on the deck of a roll-on/roll-off ship is another example on a horizontal stack.

We consider a stacking area with $s$ stacks with capacity $h$. The capacity of a stack is the upper limit on the number of items that can be stored in the stack at any given point in time. A load/unload point is a point in time or space (for example a port) where items can enter and leave the stacking area. We focus on the {\em discrete} case of the stacking problem with only a few possible load/unload points as opposed to the {\em continuous} case where any real number can be a load/unload point.  We restrict our attention to homogeneous items where any item can be stacked on top of another item. The objective for the stacking problem is to minimize the number of {\em shifts} where a shift refers to a situation where you have to {\em restow} an item $a$ to get access to another item $b$ in the same stack as $a$ because $b$ has to leave the stacking area before $a$.

For some stacking scenarios, you do not have full information on what items you have to deal with in the future. In such a case you are looking for {\em online} algorithms that take a decision on where to stack an item when the item arrives only using information on items already present in the stacking area. An algorithm where all information on future items is available is called an {\em offline} algorithm.

\subsection{Related Work}

The offline case of the stacking problem is NP-hard for any fixed value of $h \geq 6$~\cite{Cornelsen2007,Jansen2003}, but the complexity of the problem for fixed values $2 \leq h \leq 5$ is an open problem (to the best of our knowledge). The offline version of the problem is also NP-hard for unbounded stack capacity~\cite{Avriel2000}. The offline problem can be solved in polynomial time if both $s$ and $h$ are fixed constants~\cite{Tierney2014} but the running time is enormous.

Olsen et al.~\cite{Olsen2023} present an online stacking algorithm that performs well asymptotically for the continuous case under certain assumptions for the distribution on the intervals for the items. The competitive ratio for online stacking has been examined by Demange et al.~\cite{Demange2012} in a train track scenario and subsequently by Demange and Olsen~\cite{DBLP:conf/walcom/DemangeO18}. It is worth mentioning the work of Rei and Pedroso~\cite{Rei13}
and K\"onig et al.~\cite{Konig2007}
on stacking problems in the context of the steel industry. Finally, it is also worth mentioning the stacking heuristics presented in~\cite{Borgman2010,Duinkerken2001,Wang2014}.

\subsection{Contribution}

Let $p$ denote the number of load/unload points and let $m$ denote the maximum number of items present at the same time in the stacking area. The capacity of the stacking area $sh$ should be at least $m$ and you can avoid shifts if you add some extra capacity in addition to $m$ to the stacking area. As an example, you can put items with the same destination in the same stack in which case you will avoid shifts using extra capacity not exceeding $ph$ (roughly). The stacking problem is computationally hard so it is not easy to figure out exactly how much extra capacity is needed to avoid shifts. 

The key focus of our paper is to present a simple online algorithm that avoids shifts using only a little extra space in the stacking area in addition to $m$. For a real number $x$, we let $\lfloor x \rfloor$ denote the largest integer $z$ such that $z \leq x$ and let $\lceil x \rceil$ denote the smallest integer $z$ with $x \leq z$. The contribution of our paper can be summarized as follows:
\begin{itemize}
\item We present a simple online stacking algorithm and a sufficient condition for producing a stacking with no shifts using the algorithm:
$$ m + (h-1)\left\lceil p/2 \right\rceil \leq sh \enspace . $$
The condition says that the algorithm avoid shifts if we use ${(h-1)\left\lceil p/2 \right\rceil}$ extra space in addition to $m$ in the stacking area. Please note that less than $p/2$ extra stacking positions are used in addition to the necessary $m/h$ positions if $p$ is even (can be seen by division with $h$).
\item If we modify the condition slightly and lower the stacking area capacity as follows
$$ m + (h-1)\left(\left\lfloor p/2 \right\rfloor-1\right) - 1 = sh \enspace ,$$
then there are stacking problem instances that satisfy the condition where shifts are unavoidable no matter what algorithm is used (even offline algorithms). This shows that our condition is relatively tight.
\item A key ingredient for our algorithm is a procedure for splitting the items into $\left\lceil p/2 \right\rceil$ groups such that items within the same group can be put in the same stack without causing shifts. This procedure might also be useful for other stacking approaches.
\end{itemize}

The preliminaries are covered in Sec.~\ref{sec:preliminaries} and the algorithm and main results of the paper are presented Sec.~\ref{sec:algorithm} followed by a conclusion.

\section{Preliminaries}\label{sec:preliminaries}

We will say that two intervals $[p_1, p_2]$ and $[q_1, q_2]$ {\em overlap} if $p_1 < q_1 < p_2 < q_2$ or $q_1 < p_1 < q_2 < p_2$. For each item that has to be stored in the stacking area there is an interval $[p_1, p_2]$ where $p_1$ and $p_2$ denote the loading and unloading point for the item, respectively. A key observation is that a shift is necessary if and only if there are two items with overlapping intervals assigned to the same stack. This is under the assumption that items loaded at the same point will be loaded such that items traveling for the longest distance will be loaded first and that items will be unloaded before new items are loaded. This natural order of loading and unloading will be assumed throughout the paper.

Please note that there might be several items represented by the same interval. Please also note that $p_1$ and $p_2$ can be points in space if the stacking area as an example is a container ship, but $p_1$ and $p_2$ can also be points in time if the stacking area as an example is a container terminal. The number of possible loading and unloading points is denoted by $p$: $p_1, p_2 \in \{0, 1, \ldots, p-1\}$.

An instance of a stacking problem is represented by the multiset of intervals for the items. The objective is to come up with an assignment of the items to the stacks in the stacking area such that 1) the capacity of the stacks is respected, and 2) the number of shifts is minimized. The definition of the problem looks as follows:
\begin{definition}
\label{def:stacking} 
  The stacking problem is defined as follows:
  \begin{itemize}
  \item 
    Input: Two positive integers $s$ and $h$ and a multiset $\mathcal{I}$ of intervals where each interval $[p_1, p_2] \in \mathcal{I}$ satisfies $p_1, p_2 \in \{0, \ldots, p-1\}$ and $p_1 \neq p_2$. The number $s$ denotes the number of stacks for a stacking area and $h$ denotes the capacity for each stack. Each interval in $\mathcal{I}$ represents an item loaded at $p_1$ and unloaded at $p_2$. 
  \item 
    Output: An assignment of the items to stacks with a minimum number of shifts.
  \end{itemize} 
\end{definition}

There is a connection between the stacking problem and graph coloring that we now will explain. Consider the graph $G_p$ with a vertex for each possible interval $[p_1, p_2]$ with $p_1, p_2 \in \{0, \ldots, p-1\}$ and $p_1 \neq p_2$ and an edge between two vertices if and only if the two corresponding intervals overlap. A subgraph of $G_7$ is depicted below in Fig.~\ref{fig:G7} (the symbols in the figure will be explained later). The graph $G_p$ and any of its subgraphs is called an {\em interval overlap graph}. A coloring -- or more precisely a proper vertex coloring -- of a graph is an assignment of colors to the vertices such that no two vertices of the same color has an edge between them. So if two vertices have the same color, then their intervals are not overlapping. In the following the color of an item will refer to the color of the interval of the item.

\section{An Online Stacking Algorithm}\label{sec:algorithm}

Now assume that we have a coloring of $G_p$ and an instance of the stacking problem. We can avoid shifts if we only assign items to the same stack if they have the same color in $G_p$. Intuitively, we obtain a good solution if only a few colors are used in the coloring. These considerations lead to the following simple algorithm for assigning an arriving item to a stack where a {\em partially filled} stack is a stack containing at least one item but strictly less than $h$ items: if there is a partially filled stack with items with the same color as the arriving item, then put the arriving item in such a stack. Otherwise, start a new stack with the arriving item at the bottom. As an example, the items could be colored according to their destination (${\mbox{color}(p_1, p_2, p) = p_2}$) using $p-1$ colors but as we will see later it is possible to use fewer colors.

The algorithm is listed in Fig.~\ref{fig:alg_stack} where the color-function returns the color of the interval $[p_1, p_2]$ in $G_p$, and where we assume that the items are presented to the algorithm in the natural order described in Sec.~\ref{sec:preliminaries}.

\begin{figure}
  \begin{center}
    \mbox{\begin{minipage}[t]{0cm}
      \begin{tabbing}
        x\=xxx\=xxx\=xxx\=xxx\=xxx\=\kill
        \textbf{Algorithm stack($p_1$, $p_2$, $p$)}: \\
        \> \ 1:\>  $c \leftarrow$ color($p_1$, $p_2$, $p$) \\
        \> \ 2:\>  \\
        \> \ 3:\>  {\bf if} there is a partially filled stack with items with color $c$ {\bf then} \\
        \> \ 4:\>\>  put the item on top of such a stack \\
        \> \ 5:\>  {\bf else} \\
        \> \ 6:\>\>  start a new stack with the item at the bottom \\
        \> \ 7:\>  {\bf end if} \\
     \end{tabbing}
   \end{minipage}}
  \end{center}
  \caption{The pseudocode for our online stacking algorithm.}
  \label{fig:alg_stack}
\end{figure}

The strategy for the stacking algorithm is similar to the strategy used for the continuous case in~~\cite{Olsen2023} coauthored by the author of this paper but the color-functions are completely different. The three key questions that are addressed in the remaining part of the paper are:
\begin{enumerate}
\item How many colors do we need to color $G_p$?
\item Is there a simple procedure to color $G_p$ with a minimum number of colors?
\item How much extra space do we need in the stacking area in addition to $m$ to guarantee no shifts using the stacking algorithm?
\end{enumerate}
First, we will work on question 1 and 2. In order to do that, we need to introduce some graph terminology.

The {\em chromatic number} of a graph is the minimum number of colors needed to color it, and a {\em clique} in a graph is a set of vertices where every pair of vertices is connected by an edge. The maximum size of a clique in a graph is called the {\em clique number} of the graph and the clique number is easily seen to be a lower bound for the chromatic number. In general, the chromatic number can be much bigger than the clique number which is also the case for interval overlap graphs: Davies~\cite{DaviesJames2022Ibfc} has shown that there are interval overlap graphs with clique number at most $\omega$ and chromatic number $\omega(\ln \omega-2)$ for any $\omega$. Wet let $\chi_p$ denote the chromatic number of $G_p$.

The following lemma provides answers for questions 1 and 2:

\begin{lemma}
\label{lem:coloring}
The chromatic number $\chi_p$ for $G_p$ is $\left\lceil p/2 \right\rceil$, $p \neq 3$, and the coloring algorithm in Fig.~\ref{fig:alg_color} produces a coloring using $\chi_p$  colors.
\end{lemma}

\begin{figure}
  \begin{center}
    \mbox{\begin{minipage}[t]{0cm}
      \begin{tabbing}
        x\=xxx\=xxx\=xxx\=xxx\=xxx\=\kill
        \textbf{Algorithm color($p_1$, $p_2$, $p$)}: \\
        \> \ 1:\>  {\bf if} $p_2 - p_1 \leq \left\lceil p/2 \right\rceil$ {\bf then}\\
        \> \ 2:\>\>  {\bf if} $p_2 \geq \left\lceil p/2 \right\rceil$ {\bf then}\\
        \> \ 3:\>\>\>  {\bf return} $p_2-\left\lceil p/2 \right\rceil$ \\
        \> \ 4:\>\>  {\bf else} \\
        \> \ 5:\>\>\>  {\bf return} $p_2$ \\
        \> \ 6:\>\>  {\bf end if} \\
        \> \ 7:\>  {\bf else} \\
        \> \ 8:\>\> {\bf return} $p_1$ \\
        \> \ 9:\>  {\bf end if} \\
     \end{tabbing}
   \end{minipage}}
  \end{center}
  \caption{This algorithm produces a coloring of $G_p$ using $\chi_p = \left\lceil p/2 \right\rceil$ colors.}
  \label{fig:alg_color}
\end{figure}

\begin{proof}
The following $\left\lfloor p/2 \right\rfloor$ intervals with length $\left\lceil p/2 \right\rceil$ form a clique in $G_p$ that we will refer to as $C_p$: $\left[0, \left\lceil p/2 \right\rceil\right]$, $\left[1,\left\lceil p/2 \right\rceil+1\right]$,  $\left[2, \left\lceil p/2 \right\rceil+2\right]$, $\ldots$,  $\left[\left\lfloor p/2 \right\rfloor-1, p - 1 \right]$. As an example, $C_7$ is a triangle as depicted in Fig.~\ref{fig:G7}. There are not bigger cliques in $G_p$ than $C_p$ since all the endpoints for intervals in a clique must be different. We thus conclude that the clique number of $G_p$ is $\left\lfloor p/2 \right\rfloor$ implying $\chi_p \geq \left\lfloor p/2 \right\rfloor$.

We now consider the algorithm color listed in Fig.~\ref{fig:alg_color}. It is not hard to see that the algorithm uses $\left\lceil p/2 \right\rceil$ colors: The set of numbers returned in line 3 is $\{0, 1, \ldots, p-1-\left\lceil p/2 \right\rceil\} = \{0, 1, \ldots, \left\lfloor p/2 \right\rfloor-1\} $, the set of numbers returned in line 5 is $\{1, 2, \ldots, \left\lceil p/2 \right\rceil-1\}$, and the set of numbers returned in line 8 is $\{0, 1, \ldots, \left\lfloor p/2 \right\rfloor-2\}$. The maximum number returned in line 8 is the color of the interval $[\left\lfloor p/2 \right\rfloor-2, p-1]$.

In order to prove that the coloring is proper, we need to prove that two intervals with the same color do not overlap. If identical colors of two intervals are returned in the same line in the algorithm, then the intervals share an endpoint in which case they are not overlapping. This is also the case if the colors are returned in lines 5 and 8.

Now assume that $[p_1, p_2]$ is colored in line 3 and $[q_1, q_2]$ is colored in line 8 with the same color implying
$$p_2 - p_1 \leq \left\lceil p/2 \right\rceil$$
$$q_2 - q_1 > \left\lceil p/2 \right\rceil$$
$$q_1 = p_2 - \left\lceil p/2 \right\rceil \enspace .$$
From the first and the third equation above we can see that $p_1 \geq q_1$. From the second and the third equation we deduce that $p_2 < q_2$. We now conclude that the intervals are not overlapping.

We now examine the case where $[p_1, p_2]$ and $[q_1, q_2]$ receive the same color in line 3 and 5, respectively. In this case we have the following
$$p_2 - p_1 \leq \left\lceil p/2 \right\rceil$$
$$q_2 = p_2 - \left\lceil p/2 \right\rceil \enspace$$
implying $p_1 \geq q_2$ which means that the intervals are not overlapping. We have now examined every case for returning identical colors and we conclude that the coloring is proper implying $\chi_p \leq \left\lceil p/2 \right\rceil$.

At this point in the proof, we know that $\left\lfloor p/2 \right\rfloor \leq \chi_p \leq \left\lceil p/2 \right\rceil$. We will use a proof by contradiction to prove that $\chi_p > (p-1)/2$ if $p$ is odd to show that $\chi_p = \left\lceil p/2 \right\rceil$. Now assume that $p$ is odd and that we have a coloring of $G_p$ using $(p-1)/2$ colors.

The following $(p-1)/2$ intervals with length $(p-1)/2$ form a clique in $G_p$: $\left[0, (p-1)/2\right]$, $\left[1, (p-1)/2+1\right]$, $\left[2, (p-1)/2+2\right]$, $\ldots$,  $\left[(p-3)/2, p-2\right]$. These intervals are colored with $(p-1)/2$ different colors. Now consider the interval $a = \left[(p-1)/2, p-1\right]$ that also has length $(p-1)/2$. This interval overlaps all the intervals in the clique except the interval $b = \left[0, (p-1)/2\right]$. This means that $a$ and $b$ must have the same color.

Any interval in the clique $C_p$ defined above will overlap at least one of the intervals $a$ or $b$, so the color used for $a$ and $b$ cannot appear in $C_p$ contradicting that only $(p-1)/2$ are used to color $G_p$. The stage in the proof where we have to color $C_p$ and obtain a contradiction is illustrated in Fig.~\ref{fig:G7} for $p=7$.
\end{proof}

\begin{figure}
\centering
\includegraphics[scale=1.2]{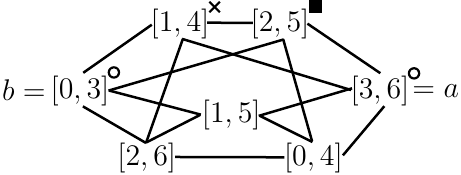}
\caption{A subgraph of $G_7$ is shown. Four of the vertices have received a color represented by a circle, cross or a square. It is not possible to color the remaining vertices in $C_7$ using only the three colors used so far (the clique $C_7$ is the triangle in the bottom and middle).}
\label{fig:G7}
\end{figure}
As a remark to the proof, a contradiction will not appear for $p=3$ consistent with the fact that $\chi_3 = 1$.
 
We can now answer question 3 from above by presenting a sufficient condition for avoiding shifts addressing how much extra space is needed in the stacking area:
\begin{theorem}
\label{thm:algorithm} The stacking algorithm in Fig.~\ref{fig:alg_stack} avoids shifts if it uses the color-function from Fig.~\ref{fig:alg_color} and
\begin{equation}
\label{eq:sufficient_condition}
m + (h-1)\left\lceil p/2 \right\rceil \leq sh \enspace .
\end{equation}
\end{theorem}

\begin{proof}
Now imagine that the stacking area is flexible with access to an infinite number of stacking positions. The strategy of the proof is to specify an upper bound on the amount of space used by the stacking algorithm in this flexible setting. In other words, we try to identify a worst case scenario. If we have access to the amount of space given by the upper bound, then we can avoid shifts.

We claim that the following invariant holds for any color $c$ at any point in time $t$: It is possible to arrange the items with color $c$ present in the stacking area at time $t$ in a sequence $i_1$, $i_2$, $\ldots$, $i_k$ such that: 
\begin{itemize}
\item The interval for item $i_x$ contains the interval of $i_{x+1}$, $1 \leq x < k$ 
\item Items stored in the same stack occupy consecutive positions in the sequence
\item There is at most one partially filled stack with items with color $c$ at time $t$, and if such a stack exists, then the items in it will have the highest indices in the sequence
\end{itemize}
The invariant is proved using induction. The invariant clearly is true before the stacking process starts. If we unload an item, then it will be the final item $i_k$ in the sequence for the particular color. If we load an item, then it can be added to the end of the sequence of the particular color because all the other items have the same color as the arriving item. It is not hard to check that all three statements in the invariant remain true whenever we load or unload an item.

Based on the invariant, we can now describe the worst case scenario: The maximum amount of space is used if there is a point in time with $m$ items in the stacking area and a partially filled stack containing only one item for each of the possible $\chi_p = \left\lceil p/2 \right\rceil$ colors as illustrated in Fig.~\ref{fig:sufficient_condition}. The inequality (\ref{eq:sufficient_condition}) thus ensures that we have enough capacity. 
\end{proof}

\begin{figure}
\centering
\includegraphics[scale=1.2]{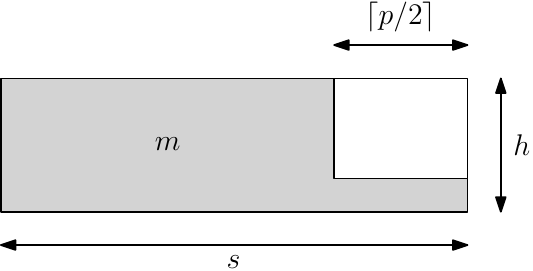}
\caption{The worst case scenario where a maximum amount of space is used in the stacking area when running the stacking algorithm.}
\label{fig:sufficient_condition}
\end{figure}

It should be noted that the color-algorithm used in the stacking algorithm could be tailored to color the {\em subgraph} of $G_p$ corresponding to the intervals in the actual instance. The factor $\left\lceil p/2 \right\rceil$ in (\ref{eq:sufficient_condition}) could then be replaced by the chromatic number of this subgraph that might be smaller. For the offline case, it could be interesting to try different color-algorithms in the stacking algorithm and choose the best of the different stacking plans produced. Please also note that there (of course) is a free choice for picking stacking positions for new stacks in the stacking algorithm. This means that constraints for stowing the items potentially can be handled by extending the algorithm with some logic for picking positions for new stacks. 

The sufficient condition from Theorem~\ref{thm:algorithm} for avoiding shifts is quite tight since we can produce stacking problem instances that use only a little less space than expressed by the condition and require shifts for {\em any} algorithm. Such instances can be constructed as follows: For each interval in the clique $C_p$ defined in the proof of Lemma~\ref{lem:coloring}, we add $1+z \cdot h$ items to the instance for some integer $z \geq 0$. We choose the $z$-values so we cannot avoid putting two items with different intervals in the same stack forcing a shift. This construction is depicted in Fig.~\ref{fig:observation}. The tightness of the sufficient condition can be expressed as follows:
\begin{observation}
\label{observation}
Stacking instances where shifts are unavoidable exist for any combination of $s, h, p, m$ satisfying
$$ m + (h-1)\left(\left\lfloor p/2 \right\rfloor-1\right) - 1 = sh \enspace .$$
\end{observation}

\begin{figure}
\centering
\includegraphics[scale=1.2]{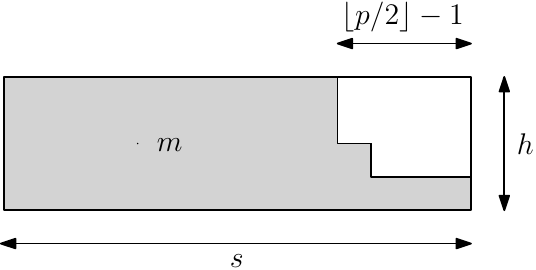}
\caption{Two items represented by two overlapping intervals are forced to be put in the same stack if the items/intervals are picked in the right way from the clique $C_p$.}
\label{fig:observation}
\end{figure}

\section{Conclusion}

We examine how much extra space it takes to avoid shifts in a stacking area. A simple online stacking algorithm is presented together with a sufficient condition guaranteeing no shifts if a specified amount of extra space is used in addition to the minimum amount needed. If the specified amount of extra space is lowered a little, then there will be stacking instances requiring shifts for any algorithm including offline algorithms (even those taking superpolynomial time). The online stacking algorithm is considered to be useful for the discrete version of the stacking problem with relatively few load/unload points. The paper also presents a procedure for coloring intervals with a minimum number of colors such that items can be stacked together with no shifts if their load/unload-intervals have the same color. This coloring procedure might also be a useful component for other stacking approaches.

There is still a gap between the extra space guaranteeing no shifts and the extra space allowing stacking problem instances with unavoidable shifts. It is an interesting problem to investigate whether this gap can be made smaller.

\bibliographystyle{plain}
\bibliography{referencer}

\end{document}